\documentclass[12pt]{article}
\usepackage[margin=1.25in]{geometry}
\usepackage[font=small,labelfont=bf]{caption}
\usepackage{amsmath,amssymb,amsthm,bm,enumitem,tikz}
\newcommand{\RS}{\text{\normalfont RS}}
\newcommand{\RM}{\text{\normalfont RM}}
\newcommand{\sh}{\text{\normalfont sh}}
\newcommand{\con}{\text{\normalfont con}}
\newcommand{\F}{\mathbb{F}}
\newcommand{\Fbar}{\overline{\mathbb{F}}}
\newtheorem{example}{Example}
\newtheorem{theorem}{Theorem}
\newtheorem{lemma}{Lemma}
\newtheorem{corollary}{Corollary}
\newtheorem{definition}{Definition}
\newcommand{\twocal}[2]{\mathcal{#1}_{\mathcal{#2}}}
\DeclareMathOperator{\wt}{wt}
\DeclareMathOperator{\dmin}{d_{\min}}
\title{\textbf{On Binary Shadow Codes}}
\author{Amir Tasbihi and Frank R. Kschischang\\
{\small Department of Electrical and Computer Engineering}\\
{\small University of Toronto}\\
{\small Emails: \texttt{\{tasbihi, frank\}@ece.utoronto.ca}}}
\date{\today}

\begin{document}
\maketitle

\begin{abstract}
\noindent
We generalize the shadow codes of Cherubini and Micheli
to include basic polynomials having arbitrary degree, and show
that restricting basic polynomials to have degree one or
less can result in improved lower bounds on the minimum distance of the code.
However, even these improved lower bounds suggest that shadow codes
have considerably inferior distance-rate characteristics compared
with the concatenation of a Reed--Solomon outer code
and a first-order Reed--Muller inner code.
\end{abstract}

\section{Introduction}

Cherubini and Micheli~\cite{cherubini2024anew} have recently introduced a new
class of low-rate binary linear codes called \emph{shadow codes}.  These codes
are obtained  by evaluating a product of basic polynomials over $\mathbb{F}_q$
(a field of odd characteristic) at a set of points where each basic polynomial
is guaranteed to be nonzero.  The resulting vector of nonzero values is mapped
to a vector over $\mathbb{F}_2$ via a homomorphism from the multiplicative
group of $\mathbb{F}_q$ to the additive group of $\mathbb{F}_2$.  Cherubini and
Micheli consider shadow codes where the basic polynomials are irreducible of
degree at least two.

In this note, we briefly discuss general properties of shadow codes and
generalize the construction to include basic polynomials having arbitrary
degree.  Restricting basic polynomials to degree one or less results in an
improved lower bound on the minimum distance of the the code.  Despite these
enhancements, those lower bounds suggest that shadow codes improve upon the
well-known Delsarte--Goethals codes \cite[Ch.~15]{MS77} only in the regime of
extremely long block lengths and in this regime the concatenation of a
Reed--Solomon outer code with a first-order Reed--Muller inner code leads to
superior code parameters.

Throughout this paper we use the following notation.  The set of
natural numbers (excluding zero)
is denoted as $\mathbb{N}$. For any $n \in \mathbb{N}$, $[n]
\triangleq \{1, \ldots, n\}$.  For a prime-power $q$, the finite field of size
$q$ is denoted with $\F_q$ and its algebraic closure is denoted as $\Fbar_q$.
Moreover, $\F_q^\ast \triangleq \F_q \setminus \{0\}$.  This
set forms a group under $\F_q$ multiplication.
The Hamming weight of a vector $\bm{v}$
is denoted as $\wt(\bm{v})$.
The minimum Hamming distance between distinct codewords
in a code $C$ is denoted as $\dmin(C)$.
For an $n \in
\mathbb{N}$, $\F_q[x_1, \ldots, x_n]$ denotes the ring of multivariate
polynomials with indeterminates $x_1, \ldots, x_n$ over $\F_q$.
For any $m \in \mathbb{N}$, the \emph{affine variety} defined
by any $m$ polynomials $P_1, \ldots, P_m \in \F_q[x_1, \ldots, x_n]$ is denoted
as $\mathcal{V}_q(P_1, \ldots, P_m)$; thus, 
\begin{equation}
\mathcal{V}_q(P_1, \ldots, P_m) \triangleq \left\{
\bm{\beta} \in \F_q^n: P_i(\bm{\beta}) = 0, \forall i \in [m]
\right\}.
\end{equation}

\section{Preliminaries}

\subsection{Absolutely Irreducible Polynomials}
For any prime-power $q$ and any $n \in \mathbb{N}$, a polynomial $P \in
\F_q[x_1, \ldots, x_n]$ is said to be \emph{irreducible} if whenever $P = AB$,
for some polynomials $A$ and $B \in \F_q[x_1, \ldots, x_n]$, then either $A$ or
$B$ is a constant field element.  Any univariate irreducible polynomial over
$\F_q$ splits (\emph{i.e.}, decomposes as a product of degree-one polynomials)
in some extension field of $\F_q$; however, the situation is
different for multivariate polynomials. Indeed, it is possible for a polynomial
$P \in \F_q[x_1, \ldots, x_n]$ to be irreducible over all algebraic extensions
of $\F_q$, \emph{i.e.}, irreducible over $\Fbar_q$.
Such a polynomial is then said to be
\emph{absolutely irreducible} over $\F_q$.

\begin{example}
The polynomial $P(x,y) = x^2 + y$ is absolutely irreducible over $\F_3$. On
the other hand, while $G(x,y) = x^2 + y^2$ is irreducible over $\F_3$, it is
not absolutely irreducible since $G(x,y) = (x - \alpha^2 y)(x +
\alpha^2 y)$, where $\alpha$ is a primitive element of $\F_9$.
\end{example}

The following lemma will be used to prove Theorem~\ref{thm:absolute
irreducibility} which provides a special class of absolutely irreducible
polynomials, extensively used throughout the paper.
\begin{lemma}
\label{lem:distinct roots}
Any irreducible polynomial in $\F_q[x]$ has distinct roots in $\Fbar_q$.
Moreover, distinct irreducible polynomials in $\F_q[x]$ do not have a common
root in $\Fbar_q$.
\end{lemma}
\begin{proof}
See~\cite[p.~520, Prop.~9]{dummit2004abstract}.
\end{proof}

\begin{theorem}
\label{thm:absolute irreducibility}
For any finite field $\F_q$ and any $m \in \mathbb{N}$, let $P_1, \ldots, P_m
\in \F_q[x]$ be distinct monic irreducible polynomials. Then, for any $\gamma
\in \F_q^\ast$, the polynomial
\[
Q(x,y) = y^2 - \gamma \prod_{i=1}^{m}P_i(x)
\]
is absolutely irreducible.
\end{theorem}
\begin{proof}
Were $Q(x,y)$ irreducible in $\Fbar_q$, 
there would be polynomials $A$ and $B \in \Fbar_q[x]$ such that 
\begin{subequations}
\begin{align}
Q(x,y) &= (y - A(x)) (y + B(x))\\
&=y^2 + (B(x) - A(x))y - A(x)B(x),
\end{align}
\end{subequations}
which implies that $A = B$. Therefore, 
\begin{equation}
\gamma \prod_{i}^{m}P_i(x) = A^2(x),
\end{equation} 
implying that all roots of $\prod_{i=1}^{m} P_i(x)$ have a multiplicity
at least two. However, since the $P_i$'s are distinct, this contradicts
Lemma~\ref{lem:distinct roots}.
\end{proof}

Somewhat
counterintuitively, while a univariate irreducible polynomial in $\F_q[x]$
does not have zeros in $\F_q$, a multivariate (absolutely) irreducible
polynomial in $\F_q[x_1, \ldots, x_n]$ may have zeros in $\F_q^n$! The
following theorem expresses bounds on the number of zeros of a class of
irreducible polynomials.
\begin{theorem}
\label{thm:number of zeros}
Let $k$ be a positive integer and $P \in \F_q[x]$ be a polynomial with $\ell$
distinct zeros in $\Fbar_q$, for some $\ell \in \mathbb{N}$, such that 
\begin{equation}
R(x,y) = y^k - P(x)
\end{equation}
is absolutely irreducible. Then, the number of zeros of $R$, \emph{i.e.},
$|\mathcal{V}_q(R)|$, satisfies
\begin{equation}
\Big| |\mathcal{V}_q(R)| - q\Big| \leq (k-1)(\ell-1)\sqrt{q}.
\end{equation}
\end{theorem}
\begin{proof}
See~\cite[Sec.~2.11]{schmidt2004equations}.
\end{proof}
\begin{corollary}
\label{cor:number of zeros}
Let $Q(x,y)$ and $P_i(x)$, $i \in [m]$, be as defined in
Theorem~\ref{thm:absolute irreducibility} and let
$d = \sum_{i \in [m]} \deg(P_i)$.
Then,
\begin{equation}
\Big| |\mathcal{V}_q(Q)| - q\Big| \leq \left(d - 1\right) \sqrt{q}.
\end{equation}
\end{corollary}
\begin{proof}
From Lemma~\ref{lem:distinct roots}, $\prod_{i=1}^{m} P_i(x)$ has $d$
distinct zeros and from Theorem~\ref{thm:absolute irreducibility}, 
$Q(x,y)$ is absolutely irreducible.  The corollary then follows from
Theorem~\ref{thm:number of zeros}.
\end{proof}

\subsection{Delsarte--Goethals Codes}

The Delsarte--Goethals (DG) codes \cite{DG75}  are a class of low-rate
binary nonlinear codes obtained as a union of carefully
selected cosets of the first-order Reed--Muller code
$RM(1,m)$ in the second-order Reed--Muller code $RM(2,m)$,
where $m$ is even.   The DG code $DG(m,d)$ (for $m=2t+2 \geq 4$)
is a code of length $2^{2t+2}$ containing
$2^{(2t+1)(t-d+2)+2t+3}$ codewords and having minimum
Hamming distance $2^{2t+1}-2^{2t+1-d}$.  The code
$DG(m,1)$ coincides with $RM(2,m)$ (the second-order
Reed--Muller code), while for $2 \leq d \leq t+1$, $DG(m,d)$
is a nonlinear subcode of $RM(2,m)$.
DG codes are well described in the classical
text of MacWilliams and Sloane \cite[Ch.~15]{MS77}.

\section{Binary Shadow Codes}

In this section, we briefly discuss binary shadow codes. For a thorough
rigorous discussion of shadow codes, see~\cite{cherubini2024anew}.

\subsection{Definition and Properties}

Let $k$ be a positive integer and $q$ be a power of an odd prime number.  Let
$\lg: \F_q^\ast \rightarrow \F_2$ be \emph{the} nontrivial homomorphism between
the multiplicative group of $\F_q$ and the additive group of $\F_2$.  Thus, for
any primitive $\alpha \in \F_q$ and for $k \in \mathbb{Z}$ we have
$\lg(\alpha^k) = k \mod 2$.  Note that a field element $\beta \in \F_q$ is a
zero of $\lg$ if and only if $\beta$ is a nonzero perfect square.

Let $\mathcal{E}$, called the \emph{evaluation set}, be any non-empty subset of
$\F_q$ and $\twocal{N}{E} \in \F_q[x]$ be the set of all polynomials which do
not vanish at any point of $\mathcal{E}$, \emph{i.e.}, 
\begin{equation}
\twocal{N}{E} \triangleq \left\{P \in \F_q[x]: P(\beta) \neq 0,
\forall \beta \in \mathcal{E}\right\}.
\end{equation}
Clearly, all irreducible polynomials of degree $\neq 1$ belong to
$\twocal{N}{E}$; therefore, $\twocal{N}{E} \neq \{\}$.  Moreover, one may see
that $\twocal{N}{E}$ forms a monoid under polynomial multiplication with the
zero-degree polynomial $1$ as its identity.

Let $\mathcal{B} \subset \twocal{N}{E}$ be any finite set such that (i) all
non-constant polynomials in $\mathcal{B}$ are monic and irreducible, and (ii)
$\mathcal{B}$ contains at most one polynomial of degree zero and
that polynomial, if contained in $\mathcal{B}$,
takes the form $P(x) = \alpha$, for a primitive $\alpha \in \F_q$.
We refer to the elements of $\mathcal{B}$ as
\emph{basic} polynomials.  Since $\F_q[x]$ is a \emph{unique factorization
domain}, each polynomial in $\F_q[x]$ can be factored uniquely as a product of
monic irreducible polynomials, multiplied by some integer power of $\alpha$.
Let $\twocal{S}{B}$ denote the set of all
polynomials whose factorization as the product of monic irreducible polynomials
has only basic factors, \emph{i.e.}, 
\begin{equation}
\twocal{S}{B} \triangleq \left\{\prod_{P \in \mathcal{B}} P(x)^{n_P} :
n_P \in \mathbb{N} \cup \{0\}\right\}.
\label{eq:sb}
\end{equation} 
Note that $\twocal{S}{B}$ is a submonoid of $\twocal{N}{E}$.

Let $\Lambda_{\mathcal{E}}: \twocal{N}{E} \to \F_2^{|\mathcal{E}|}$ be a
function defined as
\begin{equation} 
\Lambda_{\mathcal{E}}(P) \triangleq \Big( \lg\big(P(\beta)\big) : \beta \in \mathcal{E}\Big), 
\end{equation}
for any $P(x) \in \twocal{N}{E}$. Note that for any two polynomials $P(x)$
and $Q(x) \in \twocal{N}{E}$ we have $\Lambda_{\mathcal{E}}(P
Q) = \Lambda_{\mathcal{E}}(P) + \Lambda_{\mathcal{E}}(Q)$; thus,
$\Lambda_{\mathcal{E}}$ is a monoid homomorphism.  

\begin{definition}
The binary shadow code $C(\mathcal{E}, \mathcal{B})$ is the image of the monoid
$\twocal{S}{B}$ under homomorphism $\Lambda_{\mathcal{E}}$, \emph{i.e.}, 
\begin{equation}
C(\mathcal{E}, \mathcal{B}) \triangleq \Lambda_\mathcal{E}(\twocal{S}{B}).
\end{equation}
\end{definition}

For any evaluation set $\mathcal{E}$ and any set $\mathcal{B}$ of basic
polynomials, let 
\begin{equation}
\Delta(\mathcal{E}, \mathcal{B}) \triangleq |\mathcal{E}| - \frac{q}{2} -
\frac{\sqrt{q}}{2} (d_{\mathcal{B}} - 1), 
\end{equation}
where $d_{\mathcal{B}} \triangleq \deg\left(\prod_{P \in \mathcal{B}}
P\right)$ is called the \emph{total degree}.   The following theorem provides some
fundamental properties of $C(\mathcal{E}, \mathcal{B})$.

\begin{theorem}[Properties of a binary shadow code]
\mbox{}
\label{thm:shadow code properties}
\begin{enumerate}[label=(\alph*)]
\item $C(\mathcal{E}, \mathcal{B})$ is a linear code of length $|\mathcal{E}|$.
\item $C(\mathcal{E}, \mathcal{B}) = \text{span}(\Lambda_{\mathcal{E}}(\mathcal{B}))$. 
\item If $\Delta(\mathcal{E}, \mathcal{B}) > 0$ then $\dim \left(C(\mathcal{E},
\mathcal{B}) \right) = |\mathcal{B}|$.
\item  $\dmin(C(\mathcal{E},\mathcal{B})) \geq \Delta(\mathcal{E},\mathcal{B})$.
\label{thm:dmin}
\end{enumerate}
\end{theorem}
\begin{proof}
\begin{enumerate}[label=(\alph*)]
\item $\twocal{S}{B}$ is a submonoid of $\twocal{N}{E}$ and
$\Lambda_{\mathcal{E}}$ is a monoid homomorphism on $\twocal{N}{E}$. Therefore,
$C(\mathcal{E}, \mathcal{B})$ is closed under binary addition of its elements. 
\item Follows from the definition of $\twocal{S}{B}$ in (\ref{eq:sb}). 
\item We show that if $\Delta(\mathcal{E}, \mathcal{B}) > 0$ then the vectors
in the spanning set $\Lambda_{\mathcal{E}}(\mathcal{B})$ are linearly
independent.  Had they been linearly dependent, then there would be a nonzero
binary vector $\bm{b} = (b_P: P \in \mathcal{B}) \in \F_2^{|\mathcal{B}|}$ such
that $\sum_{P \in \mathcal{B}} b_P \Lambda_{\mathcal{E}}(P) = \bm{0}$, where
$\bm{0}$ denotes the zero vector. Therefore, the polynomial $A(x) = \prod_{P
\in \mathcal{B}} P(x)^{b_P}$ evaluates to nonzero perfect squares at all points
of $\mathcal{E}$.  Assume that for an $\eta \in \mathcal{E}$, $A(\eta) =
\beta^2$ for some $\beta \in \F_q^\ast$.  Then, the polynomial $Q(x,y) = y^2 -
A(x)$ has zeros at $(\eta, \beta)$ and $(\eta, -\beta)$.  Indeed, each
evaluation point contributes twice in the zeros of $Q$.  As a result, $Q$ has
at least $2|\mathcal{E}|$ roots in $\F_q^2$, \emph{i.e.}, 
\begin{equation}
|\mathcal{V}_q(Q)| \geq 2 |\mathcal{E}|.
\label{eq:VQ lower bound}
\end{equation}
Due to Theorem~\ref{thm:absolute irreducibility}, $Q$ is absolutely
irreducible; thus, Corollary~\ref{cor:number of zeros} implies that  
\begin{equation}
|\mathcal{V}_q(Q)| \leq q + (d_{\mathcal{B}} - 1) \sqrt{q}.
\label{eq:VQ upper bound}
\end{equation}
From (\ref{eq:VQ lower bound}) and (\ref{eq:VQ upper bound}) one concludes that
$\Delta(\mathcal{E}, \mathcal{B}) \leq 0$ which is a contradiction. Thus,
the vectors in the spanning set $\Lambda_{\mathcal{E}}(\mathcal{B})$ are
linearly independent and $\dim(C(\mathcal{E}, \mathcal{B})) = |\mathcal{B}|$.
\item For any nonzero codeword $\bm{v}$, there is a nonzero binary vector
$\bm{b} = (b_P: P \in \mathcal{B}) \in \F_2^{|\mathcal{B}|}$ such that $\bm{v}
= \sum_{P \in \mathcal{B}} b_P \Lambda_{\mathcal{E}}(P)$.  Let $m_{\bm{v}}$
denote the number of zero entries of $\bm{v}$. Using similar arguments as in
the proof of part (c), we see that the absolutely irreducible polynomial
\begin{equation}
Q(x,y) = y^2 - \prod_{P \in \mathcal{B}} P(x)^{b_P}
\label{eq:Qxy}
\end{equation}
 has at least
$2m_{\bm{v}}$ zeros in $\F_q^2$. Therefore, (\ref{eq:VQ upper bound}) implies
that
\begin{equation}
m_{\bm{v}} \leq \frac{q}{2} + \frac{\sqrt{q}}{2} (d_{\mathcal{B}} - 1).
\label{eq:nv upper bound}
\end{equation}
Note that $\wt(\bm{v}) = |\mathcal{E}| - m_{\bm{v}}$; therefore, 
for any nonzero codeword $\bm{v}$, $\wt(\bm{v}) \geq \Delta(\mathcal{E}, \mathcal{B})$.\qedhere
\end{enumerate}
\end{proof}

\subsection{Shadow Codes of Degree 2}

Cherubini and Micheli in~\cite{cherubini2024anew} suggest using
a set $\mathcal{B}_2$ of degree-2 irreducible basic polynomials.
We refer to that scheme as a \emph{shadow code
of degree 2}. Since there is no linear basic polynomial in a shadow code of
degree 2, we may allow $\mathcal{E} = \F_q$.  In that case, the length of
the code is $n = q$, $d_{\mathcal{B}_2} = 2 |\mathcal{B}_2|$, and
if $\Delta(\F_q,
\mathcal{B}_2) > 0$ then the dimension of the code is $k = |\mathcal{B}_2|$.
Therefore,
$\dmin(C(\mathbb{F}_q,\mathcal{B}_2))$
satisfies 
\begin{equation}
\dmin(C(\mathbb{F}_q,\mathcal{B}_2))
 \geq \Delta(\F_q, \mathcal{B}_2)
 = \frac{n}{2} - \frac{\sqrt{n}}{2} (2k - 1)
\end{equation}
and the positivity of $\Delta(\F_q, \mathcal{B}_2)$ implies that 
\begin{equation}
k \leq \left\lceil \frac{\sqrt{n} - 1}{2}\right\rceil.
\end{equation}

\section{Shadow Codes of Degree at Most 1}

Corollary~\ref{cor:number of zeros} suggests that a smaller total degree
$d_{\mathcal{B}}$ tightens the lower bound on
$\dmin(C(\mathcal{E},\mathcal{B}))$, as discussed in
Theorem~\ref{thm:number of zeros}-\ref{thm:dmin}. Therefore, by choosing the
basic polynomials among linear irreducible polynomials we may get a larger
lower bound on minimum distance.

For any evaluation set $\mathcal{E} \subset \F_q$, let the set of basic
polynomials be selected as 
\begin{equation} \mathcal{B}_1 = \left\{(x - \lambda):
\lambda \in \F_q \setminus \mathcal{E}\right\} \cup \{\alpha\}.  
\end{equation}
In this case, we refer to  the code $C(\mathcal{E}, \mathcal{B}_1)$ as a
\emph{shadow code of degree at most 1}. The inclusion of $\alpha$ in $\mathcal{B}_1$ not
only increases the dimension by one, but it means
that the all-ones word is a codeword (thus the complement of every
codeword is a codeword).

Assume that for a positive $\gamma < 1$,  $|\mathcal{E}| = \gamma q$ and, as a
result, $|\mathcal{B}_1| = (1-\gamma) q + 1$. Therefore, the length of the code is 
$n = \gamma q$, the total degree is $d_{\mathcal{B}_1} = (1 - \gamma)q$, and if 
$\Delta(\mathcal{E}, \mathcal{B}_1) > 0$ then the dimension of the code is 
$k = (1 - \gamma) q + 1$. In such a case, one may see that
$q = n + k -1$ and $\gamma = n/(n+k-1)$.  We then have
\begin{equation}
\dmin(C(\mathcal{E},\mathcal{B}_1)) 
\geq \Delta(\mathcal{E}, \mathcal{B}_1) = \frac{n - k + 1}{2}
- \frac{\sqrt{n + k - 1}}{2}(k-2).
\label{eq:shadow min distance}
\end{equation}
The following theorem expresses a bound on $k$ which guarantees the positivity of 
$\Delta(\mathcal{E}, \mathcal{B}_1)$.
\begin{theorem}
\label{thm:k sqrt n}
Consider $C(\mathcal{E}, \mathcal{B}_1)$, a shadow code of degree at most
$1$ and length $n \geq 3$.  If $k \leq \sqrt{n} + 0.5$ then $\Delta(\mathcal{E},
\mathcal{B}_1) > 0$.
\end{theorem}
\begin{proof}
Let the polynomial $S(n,k)$ be defined as 
\begin{equation}
S(n,k) \triangleq k^3 + (n-6)k^2 + (10-2n)k + (2n - 5 - n^2).
\end{equation} 
One
may see that $\Delta(\mathcal{E}, \mathcal{B}_1) > 0$ if and only if
$S(n,k) < 0$. Note that for a fixed $n$, $S(n,k)$ is a cubic polynomial in $k$.
Moreover, $S(n,0) < 0$ and $S(n,n) > 0$ for $n \geq 3$. Thus,
for a fixed $n$, $S(n,k)$ has a positive real zero $k_0(n) \in (0,n)$.
Using \emph{Cardano's method}~\cite[Sec.~14.7]{dummit2004abstract}, 
one may see that 
\begin{equation}
k_0(n) = \sqrt[3]{\xi(n) + \omega(n)} + \sqrt[3]{\xi(n) - \omega(n)}
-\frac{n-6}{3}
\label{eq:k0n}
\end{equation}
where
\begin{equation}
\xi(n) = \frac{-2n^3 + 45n^2 - 72n + 27}{54},
\end{equation}
and 
\begin{equation}
\omega(n) = \frac{(n-1)\sqrt{-12n^3 + 177n^2 - 174 n - 15}}{18}.
\end{equation}
Note that although $\omega(n)$ might assume complex values, $k_0(n)$ is a
positive real number.  Thus, for $n \geq 3$, $S(k,n) < 0$ if $k < k_0(n)$.
A very tight lower bound on $k_0(n)$ which might even be considered as an
approximation for $k_0(n)$ is $\sqrt{n} + 0.5$.
Fig.~\ref{fig:k0n} shows the $k_0(n)$ function and its approximation.
\end{proof}

\begin{figure}
\centering
\includegraphics[scale=0.6666666]{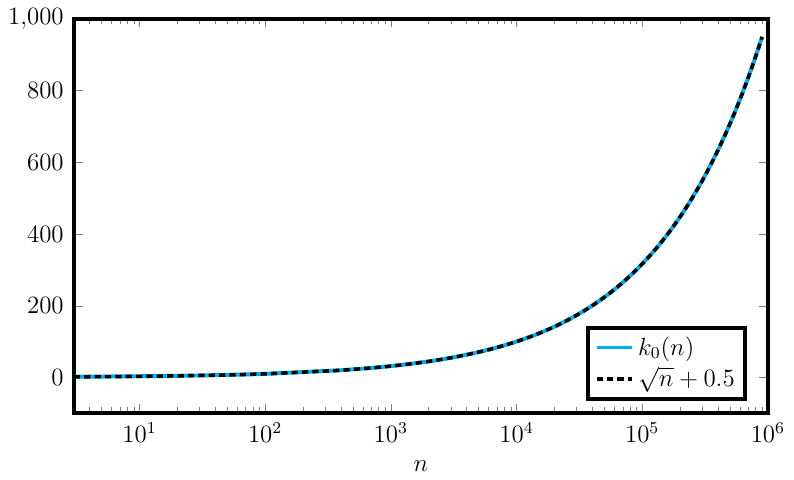}
\caption{$k_0(n)$, the upper bound on the code dimension in a shadow code 
of degree at most 1 to have 
$\Delta(\mathcal{E}, \mathcal{B}_1) > 0$, and its approximation.}
\label{fig:k0n}
\end{figure}

\section{RS-RM Concatenation}

Fix a positive integer $m$.
A natural way to construct a low-rate binary coding scheme is to
concatenate an outer $(N,K)$ Reed--Solomon (RS) code
over $\F_{2^{m+1}}$ with an
inner binary $(2^m,m+1)$ first-order Reed--Muller (RM) code.
As shown in Fig.~\ref{fig:RSRM},
each of the $N$ symbols of an RS codeword
is mapped to a single RM codeword of length $2^m$, resulting in a binary
codeword of length $N 2^m$.

\begin{figure}[htbp]
\centering
\includegraphics{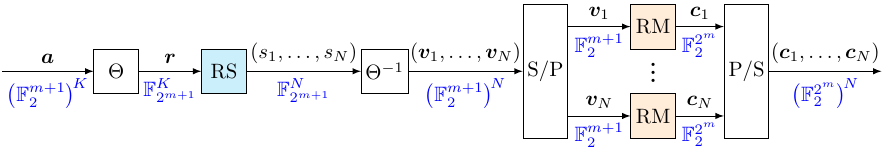}
\caption{Encoder for a concatenated RS-RM code. Above each
arrow is a vector belonging to the vector space written beneath the arrow.
The block labeled RS is a linear RS encoder, and the blocks
labeled RM are linear RM encoders.
The S/P and P/S units are serial-to-parallel and parallel-to-serial converters, respectively.}
\label{fig:RSRM}
\end{figure}

More precisely,
recall that
$\mathbb{F}_{2^{m+1}}$ is an $(m+1)$-dimensional vector space over
$\mathbb{F}_2$.  Let $\Theta: \mathbb{F}_{2}^{m+1} \to \mathbb{F}_{2^{m+1}}$
be any vector space isomorphism, and let
$\Theta^{-1}: \mathbb{F}_{2^{m+1}} \to \mathbb{F}_{2}^{m+1}$ be its
inverse.  For example,
if $\alpha$ is a primitive element of $\mathbb{F}_{2^{m+1}}$,
the map $\Theta$ taking a vector $\bm{v} = (v_1,\ldots,v_{m+1}) \in \mathbb{F}_2^{m+1}$
to
$\Theta(\bm{v}) = \sum_{i=1}^{m+1} v_i \alpha^{i-1} \in
\mathbb{F}_{2^{m+1}}$ is such an isomorphism.
We refer to $\bm{v}$ as the \emph{coordinate vector} of $\Theta(\bm{v})$.

For any $N \in [2^{m+1}]$ and any $K \in [N]$, let $\mathcal{C}_{\RS}(N,K)$ be an
RS code of length $N$, dimension $K$, and minimum Hamming distance $N-K+1$ 
over $\mathbb{F}_{2^{m+1}}$.  For any message
vector $\bm{r} \in \mathbb{F}_{2^{m+1}}^K$, let $\bm{s} = (s_1, \ldots, s_N)
\in \mathcal{C}_{\RS}(N,K)$ denote its corresponding codeword.  We
assume that the RS encoding map is linear.
Let $\bm{v}_i  = \Theta^{-1}(s_i) \in
\mathbb{F}_2^{m+1}$ denote the coordinate vector of $s_i \in
\mathbb{F}_{2^{m+1}}$, for $i \in [N]$.
Each coordinate vector $\bm{v}_i$ serves as
as a message vector for RM$(1,m)$, a binary first-order RM code of length
$2^m$, dimension $m+1$, and minimum distance $2^{m-1}$.
We assume that the RM encoding map is linear.
We denote the 
overall concatenated code by $\mathcal{C}_{\con}(N,K,m)$.

\begin{theorem}
\label{thm:rmrs}
$\mathcal{C}_{\con}(N,K,m)$ is a binary linear $(N2^m, K(m+1))$ code with 
\begin{equation}
\dmin(\mathcal{C}_{\con}(N,K,m)) \geq (N-K+1) 2^{m-1}.
\label{eq:dmin concatenate}
\end{equation}
\end{theorem}
\begin{proof}
The mapping between elements of $\mathbb{F}_{2^{m+1}}$ and vectors over
$\mathbb{F}_2^{m+1}$ is linear. As a result
the overall code is linear. The dimension and the length of the overall code 
are clear due to the described structure. 

Let $\bm{s}$ be a nonzero codeword of $\mathcal{C}_{\RS}(N,K)$.  
Since $\Theta$ is an isomorphism,
the zero of
$\mathbb{F}_{2^{m+1}}$ has the all-zero coordinate vector over $\mathbb{F}_2$.
Thus, each nonzero entry of $\bm{s}$, say $s_i \in
\mathbb{F}_{2^{m+1}}$, will be mapped to a nonzero RM codeword
$\bm{c}_i$. As a result, the corresponding codeword $\bm{c} =
(\bm{c_1}, \ldots, \bm{c}_N)$ has a weight 
\[\wt(\bm{c}) \geq \wt(\bm{s}) \dmin(\RM(1,m)) \geq (N-K+1) 2^{m-1},\]
which implies (\ref{eq:dmin concatenate}). 
\end{proof}
\begin{corollary}
Let $R_{\RS} =
K/N$ denote the rate of $\mathcal{C}_{\RS}(N,K)$. Then, the rate $R_{\con}$ and the
relative distance $\delta_{\con}$ of $\mathcal{C}(N,K,m)$ satisfy
\begin{equation}
R_{\con} = R_{\RS} \frac{m+1}{2^m},
\label{eq:rcon}
\end{equation}
and 
\begin{equation}
\delta_{\con} \geq \frac{1 - R_{\RS}}{2}.
\label{eq:con relative distance}
\end{equation}
\end{corollary}
\begin{proof}
The rate is clear. For the relative distance we have 
\[\delta_{\con} \geq 
\frac{(N-K+1) 2^{m-1}}{N2^m} = \frac{N - K + 1}{2 N} \geq \frac{N
- K}{2 N} = \frac{1 - R_{\RS}}{2}. \qedhere \]
\end{proof}

\section{RS-RM Concatenation versus Shadow Codes}

In this section, we compare the lower bounds
on relative distance of the proposed concatenated RS-RM
code and a shadow code of degree $\leq 1$. To have a fair
comparison, we assume that both codes have the same length $n$ and rate $R$.
Note that for a given $n$, in order to have an RS--RM concatenation
we must have $2^{2m+1} \geq n$. This sets a lower bound on $m$.
We choose $\mathcal{C}_{\RS}(N,K)$ and $\RM(1,m)$ codes 
in the concatenated scheme such
that their lengths over their corresponding fields agree, \emph{i.e.}, $N =
2^m$. In this case, $n = 4^{m}$. 

According to Theorem~\ref{thm:k sqrt n}, in order to use the bound
(\ref{eq:shadow min distance}) on $\dmin(C(\mathcal{E}, \mathcal{B}_1))$ the
dimension of the code, $k$, must satisfy $k \leq 2^m$. Therefore,
(\ref{eq:rcon}) implies that
\begin{equation}
R_{\RS} \leq \frac{1}{m+1}.
\label{eq:rho upper bound}
\end{equation}
From (\ref{eq:shadow min distance}), the relative distance of 
the shadow code $C(\mathcal{E}, \mathcal{B}_1)$ satisfies
\begin{subequations}
\begin{align}
\delta_{\sh} &\geq \frac{1 - R + 4^{-m}}{2} -\frac{2^m\sqrt{1 + R - 4^{-m}}}{2}(R - 2^{-2m+1})\\
&\simeq \frac{1}{2}- \frac{2^m}{2}(R - 2^{-2m+1})\\
&=\frac{1-R_{\RS}(m+1) + 2^{-m+1}}{2}.\label{eq:deltash lb}
\end{align}
\end{subequations}
By choosing $R_{\RS}$ close enough to the upper bound in
(\ref{eq:rho upper bound}) we may further simplify (\ref{eq:deltash lb}) as
\begin{equation}
\delta_{\sh} \geq \frac{1 - R_{\RS}(m+1)}{2}.
\label{eq:shadow relative distance}
\end{equation}
Comparing (\ref{eq:shadow relative distance}) with (\ref{eq:con relative distance})
implies that the lower bound on $\delta_{\con}$ is larger than the one on
$\delta_{\sh}$.
Appendix~\ref{sec:asymptotics} considers the behaviour of these
bounds in the region where $\delta \to \frac{1}{2}$
as $n \to \infty$, which is the regime for which shadow codes
were originally intended \cite{cherubini2024anew}.

\section{Results}

Figure~\ref{fig:results} compares the relative distance versus rate of binary
shadow codes, RS--RM concatenated code, Delsarte--Goethals codes, and first-
and second-order Reed--Muller codes.  For shadow codes, the indicated lines
represent lower bounds on the achievable relative distance.
It is clear from the figure that
restricting basic polynomials to have degree at most one results in an
improved lower bound on minimum distance compared with shadow codes having basic polynomials of degree two.
However, the lower bounds
for both schemes are considerably inferior to the RS--RM concatenated
code.
Neither of the shadow code lower bounds provides an improvement over
Delsarte--Goethals codes, except at extremely long block lengths.  

For $n=2^{10}$ we also computed parameters of some random linear
codes and some constructed shadow codes with basic polynomials
having degree at most one.  The results show that, at
least for this block length, the lower bound on minimum distance
for shadow codes is far from being tight.   Further work is probably
needed to provide improved bounds (or actual minimum distances).
Further work would also be needed to develop an efficient decoding
algorithm for shadow codes.

\begin{figure}
\mbox{}\hspace*{-1cm}\begin{tabular}{cc}
\includegraphics[scale=0.69]{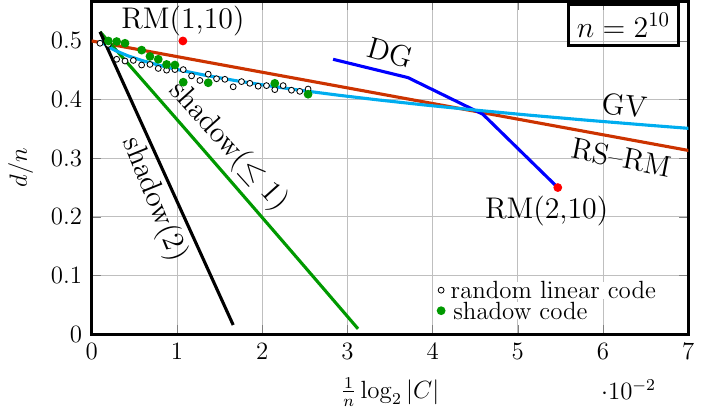}  &
\includegraphics[scale=0.69]{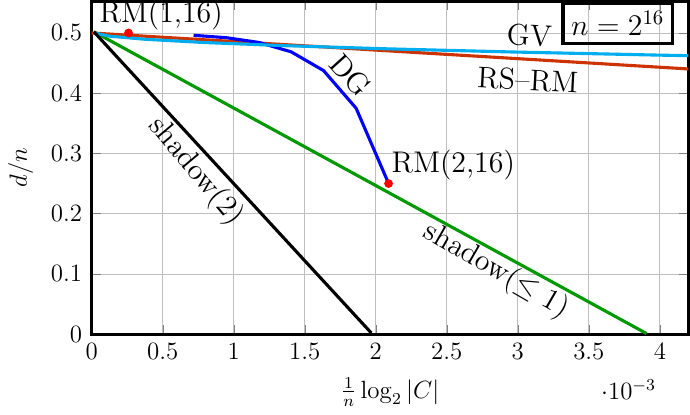}   \\
\includegraphics[scale=0.69]{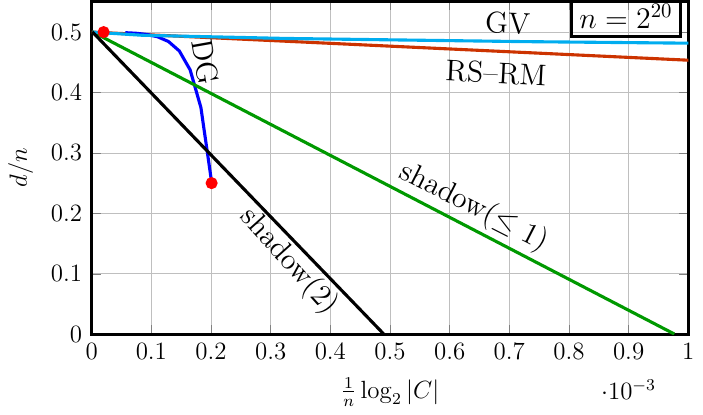}  &
\includegraphics[scale=0.69]{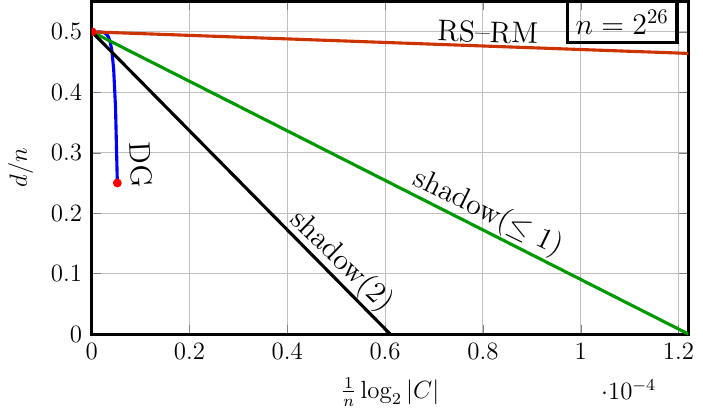} 
\end{tabular}
\caption{Comparing binary shadow codes of degree at most 1 and degree 2, with
the RS--RM concatenated, Delsarte--Goethals (DG), and first-order RM codes of the same
length.  The vertical axis denotes relative distance and the
horizontal axis denotes the code rate.   Also shown is the Gilbert-Varshamov
(GV) bound~\cite[Ch.~1, Thm.~12]{MS77}.
 For shadow codes, the indicated
lines represent lower bounds on relative distance.
The first- and second-order RM codes are
shown with a single dot.  For $n=2^{10}$, actual code parameters for
shadow codes and some randomly generated linear codes are also shown.}
\label{fig:results}
\end{figure}

\appendix
\section{The Regime where $\delta \to \frac{1}{2}$ as $n \to \infty$}
\label{sec:asymptotics}

From (\ref{eq:shadow min distance}) we know that the relative
minimum Hamming distance $\delta_{\sh}(n,k)$ of an $(n,k)$
shadow code of degree at most 1 satisfies
\[
\delta_{\sh}(n,k) \geq \frac{1}{2} -
\frac{k-1 + (k-2)\sqrt{n + k-1}}{2n},
\]
provided that the right-hand side is positive.
The latter condition is guaranteed by Theorem~\ref{thm:k sqrt n}
to hold when $k \leq \sqrt{n} + 0.5$.
Suppose we select $k = n^a$ for some $a \in (0,1/2]$. 
Since, in the limit as $n \to \infty$, we have
\[
\frac{n^a-1 + (n^a-2)\sqrt{n + n^a -1}}{2n} \to 
\begin{cases} 0 & \text{if }a < \frac{1}{2}; \\
             \frac{1}{2} & \text{if }a = \frac{1}{2},
\end{cases}
\]
we see that in the same limit
\[
\delta_{\sh}(n,n^a) \to
\begin{cases} \frac{1}{2} & \text{if }a < \frac{1}{2}, \\
                                            0  & \text{if }a = \frac{1}{2}.
\end{cases}
\]      
Thus, for example,  the family of $(n,n^{0.49})$ shadow codes
of degree at most one
has relative distance converging to $\frac{1}{2}$ as $n \to \infty$.
Putting $k = n^{1/2-\epsilon}$ for $\epsilon \in (0,1/2)$
gives
$\delta_{\sh}(n,n^{1/2-\epsilon}) = \frac{1}{2} + \mathcal{O}(n^{-\epsilon})$,
as shown in \cite{cherubini2024anew}.

In contrast, a first-order Reed--Muller code of length $n=2^m$
has relative distance exactly $\frac{1}{2}$, but a much
smaller dimension $k = 1+m = 1 + \log_2(n)$.
The Delsarte--Goethals code  $DG(m=2t+2,d)$, with $2 \leq d \leq t+1$
has relative distance
$\frac{1}{2} - \frac{1}{2^{d+1}}$ which approaches $1/2$ as
$t \rightarrow \infty$ provided that $d$ is any increasing function
of $t$.  However, the number of codewords is never more than
that of the second-order Reed-Muller code of length $n=2^{m}$,
which has dimension $1 + \binom{m}{1} + \binom{m}{2}$, which scales
as a quadratic polynomial in $m=\log_2(n)$.
Thus shadow codes are far better than the first-order Reed--Muller
and Delsarte--Goethals in the regime considered in this appendix.

However,
the concatenation of a $(2^m,2^m-K,2^m-K+1)$ Reed--Solomon code
over $\mathbb{F}_{2^{m+1}}$ with a binary
$(2^m, m+1, 2^{m-1})$
first-order Reed--Muller code gives a binary
$(2^{2m}, K(m+1), (2^m-K+1)2^{m-1})$ code having relative distance
\[
\delta_{\con}(m,K) = \frac{1}{2} - \frac{K-1}{2^{m+1}}.
\]
In terms of the parameters $n = 2^{2m}$ and $k = K(m+1)$, we have
\begin{equation}
\delta_{\con}(n,k) = \frac{1}{2} - \frac{k - \log_2\sqrt{n} - 1}{\sqrt{n} (\log_2(n) + 2)}.
\label{eqn:deltacon}
\end{equation}
Suppose we set $k = n^a \log_2(n)$ for some $a \in (0,1/2]$.
We see, in the limit as $n \to \infty$, that
\[
\delta_{\con}(n,n^a \log_2(n)) \to 
\begin{cases} \frac{1}{2} & \text{if }a < \frac{1}{2}, \\
                                            0  & \text{if }a = \frac{1}{2}.
\end{cases}
\]
Thus, for example, the family of $(n,n^{0.49}\log_2(n))$ concatenated
RS-RM codes
has relative distance converging to $\frac{1}{2}$ as $n \to \infty$.

According to (\ref{eqn:deltacon}),
to achieve relative distance $\delta_{con}(n,k) = \frac{1}{2} - \epsilon$
we must set
$k = \epsilon\sqrt{n} (\log_2 n + 2) + \log_2 \sqrt{n} + 1$,
thus
\[
\frac{k}{n} = \frac{\epsilon( \log_2 n + 2)}{\sqrt{n}}
+ \frac{ \log_2 \sqrt{n} + 1}{n} = \Omega \left( \frac{\epsilon \log_2 n}{\sqrt{n}} \right).
\]
In other words, for a fixed
$\epsilon > 0$, the concatenated RS-RM codes achieve
relative distance $\delta_{\con} = 1 -\epsilon$ with
code rate $R_{\con}$ scaling as $\Omega(\epsilon n^{-1/2} \log_2 n)$.
If $\epsilon$ is scaled with $n$ as $\epsilon = n^{-\alpha}$
for $\alpha \in (0,1/2)$,
then the concatenated RS-RM codes achieve
relative distance $\frac{1}{2} - n^{-\alpha}$ with
rate $\Omega(n^{-(\frac{1}{2} + \alpha)} \log_2 n)$.

To make a direct comparison between concatenated RS-RM and shadow codes, we can scale
their rate according to the same function of block length.
Fig.~\ref{fig:k is root n} shows the relative distance vs.~rate for concatenated RS-RM concatenated
code and the shadow codes of degree at most $1$, both having
dimension $k = n^{0.49}$. 
Although both schemes approach a relative distance of $1/2$ at large block
lengths, concatenated RS-RM code always have a greater relative distance guarantee.  We must emphasize again, however, that the \emph{actual} relative
distance of shadow codes may be greater than the indicated lower bound.
\begin{figure}[b]
\centering
\includegraphics[scale=0.72]{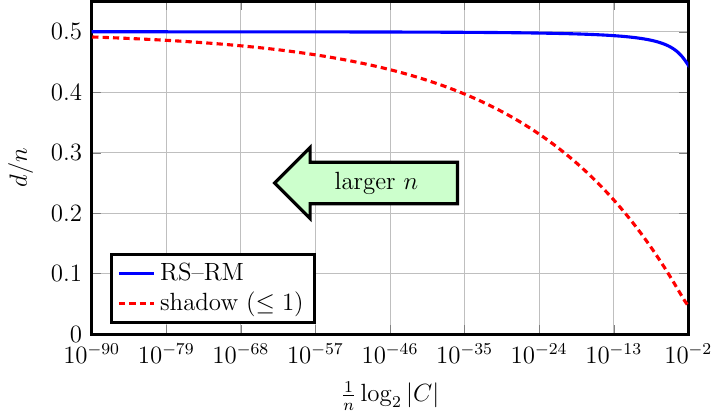}
\caption{Bounds on relative distance versus rate when
both code families have dimension $k$ scaling with
block length $n$ as $k = n^{0.49}$.}
\label{fig:k is root n}
\end{figure}

\bibliographystyle{IEEEtran}
\bibliography{IEEEabrv,ref}

\end{document}